\documentclass{article}

\usepackage{makeidx}        
\usepackage{graphicx}        
                             
\usepackage{multicol}        
\usepackage[bottom]{footmisc}

\usepackage{paralist}
\usepackage{amsmath}
\usepackage{amssymb}
\usepackage{color} 
\usepackage{flafter}
\usepackage[width=0.9\linewidth]{caption}
\usepackage[vlined,ruled,titlenumbered]{algorithm2e}

\usepackage{hyperref}

\usepackage{amsthm}

\newtheorem*{definition*}{Definition}
\newtheorem{lemma}{Lemma}
\newtheorem{theorem}{Theorem}

\newtheorem*{proposition*}{Proposition}
\newcommand\ex{\mathrm{E}}
\newcommand\prob{\mathrm{P}}
\newcommand\var{\mathrm{Var}}
\newcommand\Real{\mathbb{R}}

\graphicspath{{.}{graphics/}}
\begin{document}

\title{Estimation of Small $s$-$t$ %
  Reliabilities in Acyclic Networks}
%\titlerunning{Monte-Carlo Estimation of $s$-$t$ Reliability in Acyclic Networks}
\author{Marco Laumanns and Rico Zenklusen \\
\small\texttt{\{marco.laumanns,rico.zenklusen\}@ifor.math.ethz.ch}\\
\small ETH Zurich, Institute for Operations Research, 8092 Zurich, Switzerland
}

\maketitle

\begin{abstract}
In the classical $s$-$t$ network reliability
problem a fixed network $G$ is given including two designated
vertices $s$ and $t$ (called terminals). The edges are
subject to independent random failure, and the task is to compute
the probability that $s$ and $t$ are connected in
the resulting network, which is known to be $\#P$-complete.
In this paper we are interested in approximating the $s$-$t$
reliability in case of a directed acyclic original
network $G$. We introduce and analyze an
specialized version of a Monte-Carlo algorithm given
by Karp and Luby. 
For the case of uniform edge failure probabilities, we
give a \mbox{worst-case} bound on the number of samples that
have to be drawn to obtain an $\epsilon-\delta$ approximation,
being sharper than the original upper bound.
We also derive a variance reduction of the
estimator which reduces the expected number of iterations
to perform to achieve the desired accuracy when
applied in conjunction with different stopping rules.
Initial computational results on two types of random networks
(directed acyclic Delaunay graphs and a slightly modified version
of a classical random graph) with up to one million vertices are presented.
These results show the advantage of the introduced Monte-Carlo approach compared
to direct simulation when small reliabilities have to
be estimated and demonstrate its applicability on \mbox{large-scale}
instances.

\end{abstract}

\section{Introduction}
In the classical $s$-$t$ network reliability problem, a
fixed graph $G$ with two special vertices $s$ and $t$
 is given whose edges fail (disappear) independently of
each other with some given probability. The task is to
determine the probability that $s$ and $t$ are still
connected in the resulting network after edge failures. A
famous related problem is the all-terminal reliability
problem, where the goal is to determine the probability that
all vertices are still connected to each other after edge
failures (for further information on various reliability
problems see  \cite{colbourn_1987_combinatorics}).
Both problems are known to be computationally hard
($\#P$-complete) even for very restricted classes of graphs $G$
\cite{valiant_1979_complexity,provan_1983_complexity}.
In particular, the $s$-$t$ reliability problem
remains hard in the case when $G$ is a directed acyclic planar graph of maximal degree
three \cite{provan_1986_complexity}. Therefore we are interested in finding
approximations.

Randomized algorithms respecting some
relative error bound with high probability have shown to be
interesting approaches for different reliability
problems
\cite{karp_1985_montecarlo,karger_1995_randomized,karger_1997_implementing}.
As we are interested in relative error
bounds, there is a significant difference between the
estimation of the failure probability of the network
and estimating the probability that the
network is intact after edge failures. This results from the
fact that estimating small probabilities by sampling is in
general much more difficult than estimating large probabilities.
Moreover, the techniques used for the estimation of reliability values
often differ significantly from those used for the estimation
of the failure probability.
Most of the literature concentrates on estimating the
failure probability because in typical applications, as
communication or electrical networks, we have a
highly reliable network and are interested in accurately estimating the
probability of rare failures.

In this paper, we concentrate on estimating the probability that the network is intact, i.e., the
probability that there is a path from $s$ to $t$ after edge
failures. This is motivated by certain models of spreading processes on networks, such as
disease spreading, which can be mapped onto reliability
problems \cite{grassberger_1983_critical,sander_2001_percolation,%
warren_2001_firewalls,newman_2002_spread,newman_2003_structure}. In this context,
a path from $s$ to $t$ represents the spread of
a disease from $s$ to $t$, and we are interested in estimating the
probability of the rare event that the disease spreads over a
long distance from $s$ to $t$ .

The basis of our approach is a method presented by Karp and Luby
\cite{karp_1985_montecarlo} for estimating the
failure probability in an $s$-$t$ reliability problem when the
graph $G$ is planar, which can be easily adapted for estimating the
probability of connectedness of $s$ and $t$ on an arbitrary graph $G$.
To be efficiently applicable, however, the underlying graph $G$
needs to fulfill some additional properties, such as very low
intactness probabilities of the edges and low vertex degrees.
Furthermore, a computationally expensive preprocessing phase is
required, which computes different quantities needed to efficiently
sample from the proposed sample space. This makes the method not
suited for application on large networks. We show in
this paper that in the case of a directed acyclic network $G$,
some of these problems can be circumvented, and the resulting
algorithm can be applied to \mbox{large-scale} instances. The 
simplifications are due to the fact that generating
an $s$-$t$ path uniformly at random in an acyclic graph can be
performed in linear time (see Section~\ref{sec:technical_details}). %
Except for very restricted classes of initial networks such
as series-parallel
networks (see \cite{satyanarayana_1985_lineartime} for
further information), no practically useful methods for estimating small
$s$-$t$ reliabilities on large networks are known. Considering directed
acyclic graphs can thus be seen as a natural next step.

The paper is organized as follows. We begin with some
preliminaries in Section~\ref{sec:preliminaries}.
In Section~\ref{sec:direct_montecarlo} the
direct \mbox{Monte-Carlo} approach is described and
its efficiency is analyzed. Section~\ref{sec:our_algorithm}
then presents our algorithm for the estimation of small
reliabilities in directed acyclic graphs up to some technical details that
are explained in Section~\ref{sec:technical_details}. In
Section~\ref{sec:nb_samples} we discuss how many samples have
to be drawn in our algorithm to obtain a good estimate of
the reliability with high probability.
Section~\ref{sec:computational_results} contains
computational results on two types of random directed acyclic networks
comparing the direct \mbox{Monte-Carlo}
approach with our algorithm and demonstrating the applicability
of our algorithm on \mbox{large-scale} instances.

\section{Preliminaries}\label{sec:preliminaries}

Let $G=(V,E,p)$ be a directed acyclic network, where
\begin{compactitem}
\item $V$ is the set of vertices (with $|V|=n$),
\item $E$ is the set of edges (with $|E|=m$), and
\item $p:E\longrightarrow [0,1]$ is a function associating
  a failure probability with every edge of the network.
\end{compactitem}

We call a network $G$ satisfying these properties an acyclic
reliability network.
Furthermore, we fix two special vertices $s,t \in V$. Every
edge $e\in E$ fails with probability
$p(e)$ independently of the others. Let $G'$ be the
resulting graph over the intact edges after the realization of the edge failures.
We say that $G'$ is
intact if it contains a path from $s$ to $t$,
otherwise we say that $G'$ is in a failed state.
Finally, the $s$-$t$ reliability $\mathit{REL}_{s,t}(G)$  of the
network $G$ is defined as the probability of $G'$ being
intact.

Let $q:E\longrightarrow [0,1]$ be the function which
associates with every edge its probability of being intact,
i.e., $q(e)=1-p(e) \;\forall e\in E$. It is sometimes easier
to look at our reliability model in a different way where edges
appear rather than disappear. In this model we would begin
with an empty network and flip a biased coin for all
potential edges $e\in E$ to determine whether they appear.

We are mainly interested in $\epsilon-\delta$~approximations
for $\mathit{REL}_{s,t}(G)$, i.e., algorithms
returning an estimate of $\mathit{REL}_{s,t}(G)$ accurate to
within a relative error of $\epsilon$ with probability
at least $1-\delta$.
To determine how many samples we have to draw in a
\mbox{Monte-Carlo} algorithm to obtain an $\epsilon-\delta$~approximation
we often refer to the Generalized \mbox{Zero-One} Estimator Theorem
introduced by Dagum et al. in \cite{dagum_2000_optimal}. The theorem is repeated
below.

\begin{theorem}[Generalized Zero-One Estimator Theorem]\label{thm:generalized_01_thm}
Let $X$ be a random variable taking values in $[0,1]$ and let
$X_1,X_2,\dots,X_N$ denote independent random variables distributed
identically to $X$. If
$\epsilon<1$ and
$N\geq 4(e-2)\ln(2/\delta)\cdot \max\{\var[X],\epsilon\ex[X]\}/(\epsilon\ex[X])^2$
then
\begin{equation*}
\prob\left[(1-\epsilon)\ex[X]\leq\frac{1}{N}\sum_{i=1}^N X_i\leq(1+\epsilon)\ex[X]\right]>1-\delta\;.
\end{equation*}
\end{theorem}

We typically use the theorem in the following form. If
$N\geq 4(e-2)\ln(2/\delta)\cdot
(1/\epsilon^2\ex[X])$ then $(\sum_{i=1}^N X_i)/N$ is an
$\epsilon-\delta$~approximation for $\ex[X]$ (using the fact
that a random variable $X$ taking values in $[0,1]$
satisfies $\ex[X]\geq\var[X]$).

\section{A direct Monte-Carlo approach}\label{sec:direct_montecarlo}

In this section we consider a simple \mbox{Monte-Carlo} approach and
show that it is efficient for sufficiently large values of
$\mathit{REL}_{s,t}(G)$, but inefficient for the estimation of small
reliabilities.
In this approach we simply flip a biased coin for every edge $e\in E$
and observe whether $s$ and $t$ are connected in the resulting graph. Let
$X$ be the random variable corresponding to this approach
where $X=1$ if the resulting network is intact and $0$ otherwise.
The random variable $X$ has thus a Bernoulli
distribution with parameter $\mathit{REL}_{s,t}(G)$
and the reliability is estimated without bias by 
generating $N$ independent realizations of $X$ and returning their
empirical mean which we denote by $Y_N$. 
A central question when using this method is how
large $N$ has to be chosen to obtain an $\epsilon-\delta$~approximation.
A direct application of Theorem~\ref{thm:generalized_01_thm}
gives the following:

\begin{theorem}\label{thm:montecarlo_approx}
$Y_N$ is an $\epsilon-\delta$~approximation of
$\mathit{REL}_{s,t}(G)$ if $N$ satisfies
\begin{equation*}
N\geq
4(e-2)\ln\left(\frac{2}{\delta}\right)\cdot\frac{1}{\epsilon^2
\mathit{REL}_{s,t}(G)}\;.
\end{equation*}
\end{theorem}

When $\mathit{REL}_{s,t}(G)$ is bounded below by
$1/poly(n,m)$, $Y_N$ is an FPRAS ($\epsilon-\delta$~approximation algorithm
with running time bounded by a polynomial in $1/\epsilon,\log(1/\delta)$ and the
input size). The difficult case is the estimation of small values, i.e.,
small reliabilities $\mathit{REL}_{s,t}(G)$. This problem
motivated the construction of the algorithm to be presented next.

\section{Monte-Carlo method for estimating small reliabilities}\label{sec:our_algorithm}

The backbone of our algorithm is an
adaption of the \mbox{Monte-Carlo} method presented in
\cite{karp_1985_montecarlo}.
Our algorithm exploits that, in a directed acyclic network,
we can easily (in linear time) calculate
the mean number of intact paths from $s$ to $t$ after the edge failures.
This value is normally a good estimate for the reliability
in highly unreliable networks as in such networks an intact
state typically contains only
a few paths from $s$ to $t$. Using a \mbox{Monte-Carlo} approach,
our algorithm then estimates the ratio between
$\mathit{REL}_{s,t}(G)$ and the mean number of intact paths from $s$ to $t$ after edge
failures. Multiplying this estimate
with the mean number of intact paths from $s$ to $t$ yields
finally an estimate for the $s$-$t$ reliability of the network.

Note that the ratio between the mean number of intact paths from $s$ to $t$ after
edges failures and the reliability to estimate, i.e., the reciprocal of the value we
estimate in our algorithm, is exactly the mean number of
intact paths from $s$ to $t$ after edge failures conditioned on the event that
the network will be intact after the failure process. One of the main problems for
the development of methods for directly estimating this ratio is the difficulty
of choosing a sample out of the pool of intact states with probability proportional to
its real appearance probability. In fact such a sampling procedure could easily be
transformed into an FPRAS for the estimation of $\mathit{REL}_{s,t}(G)$
using techniques presented by Jerrum et al. \cite{jerrum_1986_random}.

\subsection{Notation}
Let $H=\{h\subseteq E\}$ be the set of all
possible states of the network $G$ after the realization of the edge failures,
where a state $h\in H$ represents the collection of intact
edges.
Furthermore, let $A\subset H$ be the set of all intact states.
For every state $h\in H$
we denote by $w(h)$ the probability that $h$ occurs after
the edge failure process, i.e., $w(h)=\prod_{e\in h}q(e)\prod_{e\in E\setminus h}p(e)$\;.
In particular, the weight of the set $A$ is the reliability
we want to estimate, $w(A)=\sum_{a\in A}w(a)=\mathit{REL}_{s,t}(G)$.

Let $\mathcal{P}$ be the set of all paths in
$G$ from $s$ to $t$. A path $\gamma$ is simply represented
by a subset of the edges $E$. For every state $h\in H$ we
denote by $\mathcal{P}(h)$ the set
of all paths from $s$ to $t$ in state $h$. A state $h\in H$
is intact if and only if we have
$\mathcal{P}(h)\neq \emptyset$.
With every path $\gamma \in \mathcal{P}$
we associate a weight $w(\gamma)$ which is the probability
that all edges on the path will be intact after the edge
failure process, $w(\gamma)=\prod_{e\in \gamma} q(e)$\;.

The mean number of intact paths from $s$ to $t$
after edge failures can thus be
written as $\sum_{\gamma \in \mathcal{P}} w(\gamma)$. In
Section~\ref{sec:technical_details} we will see how this quantity
can be efficiently calculated in acyclic graphs.

\subsection{Estimating the ratio between
\texorpdfstring{{\mathversion{bold}%
$\mathit{REL}_{s,t}(G)$}}{$\mathit{REL}_{s,t}(G)$} and the mean number of intact paths from
\texorpdfstring{{\mathversion{bold}$s$} to {\mathversion{bold}$t$}}{$s$ to $t$} after edge
failures}

%\subsection{Estimating the ratio between 
%{\mathversion{bold}%
%$\mathit{REL}_{s,t}(G)$} and the mean number of intact paths from
%{\mathversion{bold}$s$} to {\mathversion{bold}$t$} after edge
%failures}

To estimate the ratio between the reliability and the mean number
of intact paths from $s$ to $t$ after edge failures we
sample out of the sample space
\begin{equation*}
\Omega = \{(\gamma, a)\in\mathcal{P}\times A \mid \gamma
\subseteq a\},
\end{equation*}
where we associate the weight $w(\gamma,a)=w(a)$ with every element $(\gamma,a)\in\Omega$.
This sample space has the following
interesting properties. On the one hand, the weight of the sample
space is exactly the mean number of intact paths from $s$ to $t$
after edge failures, i.e.,
\begin{equation*}
w(\Omega)=\underset{(\gamma,a)\in\Omega}{\sum}w(\gamma,a)
=\underset{\gamma\in \mathcal{P}}{\sum}\;
\underset{w(\gamma)}{\underbrace{\underset{a\in A: \gamma \subseteq a}{\sum}w(\gamma,a)}}
=\underset{\gamma \in \mathcal{P}}{\sum}w(\gamma).
\end{equation*}
On the other hand, it is easy to sample
elements of $\Omega$ with probability proportional to their
weights by the following two-step procedure.
In a first step, a path $\gamma \in\mathcal{P}$ is drawn
with probability proportional to its weight $w(\gamma)$. How this
can be done efficiently will be discussed in
Section~\ref{sec:technical_details}. In the second step,
all edges not contained in $\gamma$ are sampled corresponding to their
appearance probabilities. The appeared edges of the second step together
with the edges in $\gamma$ form an intact state $a\in A$. The tuple
$(\gamma,a)$ is finally the sampled state.

We will now introduce influence values for the elements in $\Omega$
such that the expected influence value of a sample of $\Omega$ is
equal to $w(A)/w(\Omega)$. The expected influence value
can then be estimated by a standard \mbox{Monte-Carlo} algorithm.

The approach taken in \cite{karp_1985_montecarlo} was to fix for every intact
state $a\in A$ an arbitrary $s$-$t$ path $\gamma_a\in
\mathcal{P}(a)$. We begin by following this approach and
introduce later on another idea to reduce the
variance of the estimator. With every sample we associate an
influence value which
is equal to one if the sample is of the form $(\gamma_a,a)$
and equal to zero otherwise. The influence value of a sample
corresponds therefore
to the realization of a Bernoulli variable with parameter
\begin{equation*}
\frac{w(\{(\gamma_a,a) \in \Omega \mid a\in A\})}{w(\Omega)}=\frac{w(A)}{w(\Omega)}
\end{equation*}
which is precisely the value we would like to estimate. By repeating the sampling procedure
$N$ times and counting the fraction of samples of the
form $(\gamma_a,a)$ we therefore obtain an unbiased estimator $\xi_N$
for $w(A)/w(\Omega)$ with variance
\begin{equation*}
\var(\xi_N)=\frac{1}{N}\frac{w(A)}{w(\Omega)}\left(1-\frac{w(A)}{w(\Omega)}\right)\;.
\end{equation*}

Another unbiased estimator $\psi_N$ for $w(A)/w(\Omega)$
with smaller variance than $\xi_N$ can be
obtained by associating an influence value
$m(\gamma,a)=1/|\mathcal{P}(a)|$ with every sample
$(\gamma,a)$, and we define
$\psi_N$ to be simply the mean of the influence values
of our samples. This
approach is particularly interesting in our case as the fact
of $G$ being acyclic allows to compute the value
$|\mathcal{P}(a)|$ in linear time (see Section~\ref{sec:technical_details})
and thus does not increase the overall worst-case complexity of our
algorithm. The reduction of the variance decreases the expected number
of iterations to perform for obtaining an $\epsilon-\delta$ approximation
when applying a stopping criterion as presented in
Section~\ref{sec:nb_samples}.

Unfortunately, we cannot guarantee some minimal decrease of the variance
of the estimator $\psi_N$ compared to $\xi_N$ as there are instances
of reliability networks where we have an arbitrarily low decrease.
On the other hand, it is easy to find instances where the variance
reduces by an arbitrarily high factor.

We finally estimate $\mathit{REL}_{s,t}(G)$ by multiplying $\psi_N$ by
the weight of the sampling space ($w(\Omega)$) which can be calculated
efficiently as shown in the next section.  Algorithm \ref{alg:montecarlo_psi}
gives a pseudocode for an implementation of the $s$-$t$ reliability
estimation based on $\psi_N$.
Note that the sampling of the edges not on the path performed in
Algorithm~\ref{alg:montecarlo_psi} in lines
\ref{algline:sampling_remaining_edges} to \ref{algline:end_sampling} can
be done more efficiently as in general we do not need to sample all edges
for calculating $|\mathcal{P}(a)|$.
In Section~\ref{sec:technical_details} we will see how this part of the algorithm can
be improved and show how to perform the remaining unspecified parts of our algorithm.
More precisely, we will discuss the following operations, where $\kappa$ represents the
time needed for generating a random number uniformly in $[0,1]$:
\begin{compactitem}
\item Sampling in $O(m\kappa)$ time a path according to line~\ref{algline:path_sampling} of
Algorithm~\ref{alg:montecarlo_psi},
\item Determining in $O(m)$ time the weight $w(\Omega)$ of the sampling space $\Omega$ used in
line~\ref{algline:w_omega} of Algorithm~\ref{alg:montecarlo_psi},
\item Sampling edges not being on the initial path $\gamma$
and calculating $|\mathcal{P}(a)|$ as needed in
lines~\ref{algline:sampling_remaining_edges} to
\ref{algline:multiplicity} of Algorithm~\ref{alg:montecarlo_psi}.
\end{compactitem}

\linesnumbered
\begin{algorithm}[H]
\SetLine
\caption{Estimation of $\mathit{REL}_{s,t}(G)$
based on $\psi_N$\label{alg:montecarlo_psi}}
\SetKwInOut{Input}{Input}
\SetKwInOut{Output}{Output}
\Input{An acyclic reliability network $G=(V,E,p)$ with two special
  terminals $s,t \in V$ and the number $N$ of iterations to perform}
\Output{An estimate for $\mathit{REL}_{s,t}(G)$}
$x=0$\;
\For{$i=1$ to $N$}%
{
$a=\emptyset$\;
Draw a random path $\gamma$ out of the set $\mathcal{P}$
with probability proportional to $w(\gamma)$\label{algline:path_sampling}\;
$a=a\cup \gamma$\;
\ForEach{$e\in E\setminus \gamma$}%
{\label{algline:sampling_remaining_edges}
  Take a sample $x_e$ of a Bernoulli variable with parameter
  $q(e)$ to determine whether the edge $e$ appears\;
  \If{$x_e=1$ (the edge appears)}%
  {
    $a=a\cup e$\;
  }
}\label{algline:end_sampling}
Determine $|\mathcal{P}(a)|$\label{algline:multiplicity}\;
$x=x+\frac{1}{|\mathcal{P}(a)|}$\;
}
$\psi_N=\frac{x}{N}$\;
Determine $w(\Omega)$\label{algline:w_omega}\;
\Return{$\psi_N\cdot w(\Omega)$}
\end{algorithm}

\section{Algorithmic details}\label{sec:technical_details}

\subsection{Sampling \texorpdfstring{{\mathversion{bold}$s$}-{\mathversion{bold}$t$}}{$s$-$t$} paths}

We begin by studying how paths can be efficiently sampled
according to line \ref{algline:path_sampling} of Algorithm~\ref{alg:montecarlo_psi}.
The idea is to start at the terminal
$s$ and then to construct a path to $t$ by successively
adding new edges. The choice of a new edge augmenting the
current partial path is done in the following way. With every
edge $(v,u)\in E$ we associate a weight $\widetilde{w}(v,u)$
which is the sum of the weights of all paths from $v$ to $t$
using edge $(v,u)$, i.e.,
\begin{equation}\label{eq:widetilde_w}
\widetilde{w}(v,u)=
\sum_{\begin{subarray}{c}
\gamma: \text{ path from $v$ to $t$}\\
    \text{with }(v,u)\in\gamma
\end{subarray}}
\;\;\underset{e\in\gamma}{\prod}q(e)\;.
\end{equation}

During the path sampling method, after a partial path from
$s$ to some vertex $v$ is constructed, we choose an
outgoing edge of vertex $v$ with probability
proportional to the weights $\widetilde{w}$. It is easy to
verify that this procedure effectively samples a path
$\gamma\in\mathcal{P}$ with probability proportional to
$w(\gamma)$ as desired. Furthermore, the edge weights
$\widetilde{w}$ can be easily computed by the following
procedure.

We suppose without loss of generality that every edge lies
on at least one path from $s$ to $t$ in $G$. All edges not
satisfying this condition can be eliminated in $O(m)$ time
in a preprocessing step. We then determine a
topological order of the vertices. By the condition mentioned above
it is clear that $t$ will be the last vertex in
the topological order. We go through the vertices in reverse topological order
and determine at each step the weights $\widetilde{w}$ of the
edges entering the current vertex. At the first step we look
at every edge $e$ incident to $t$ and determine its weight
$\widetilde{w}(e)$ which is equal to $q(e)$. The weights of
the other edges can then be determined in linear time
by using the following recursive
formula which follows from the definition of the weights
$\widetilde{w}$ in (\ref{eq:widetilde_w}),

\begin{equation}\label{eq:recursive_widetilde_w}
\widetilde{w}(v,u)=q(v,u)\cdot \underset{e\in\delta^+(u)}{\sum}\widetilde{w}(e) \, ,
\end{equation}
\noindent
where for $v\in V$, $\delta^+(v)$ (resp. $\delta^-(v)$) denotes the set of all
edges going out of $v$ (resp. all edges entering $v$).

\subsection{Determining \texorpdfstring{{\mathversion{bold}$w(\Omega)$}}{$w(\Omega)$}}

Another subproblem in Algorithm~\ref{alg:montecarlo_psi} is the
calculation of $w(\Omega)$, the weight of the sample
space. Having already calculated the weight function
$\widetilde{w}$, this problem can easily be solved by
expressing $w(\Omega)$ in terms of $\widetilde{w}$
as follows.

\begin{equation*}
w(\Omega)=\sum_{\gamma \in
  \mathcal{P}}w(\gamma)=\sum_{e\in\delta^+(s)}\;%
\underset{=\widetilde{w}(e)}{\underbrace{\sum_{%
\substack{\gamma\in\mathcal{P}\\e\in\gamma}}w(\gamma)}}=%
\sum_{e\in\delta^+(s)}\widetilde{w}(e)
\end{equation*}

\subsection{Sampling edges outside the initial path and calculating
\texorpdfstring{{\mathversion{bold}$|\mathcal{P}(a)|$}}{$|\mathcal{P}(a)|$}}\label{subsec:sampling_remaining_edges}

As the sampling of the remaining edges in
lines~\ref{algline:sampling_remaining_edges} to
\ref{algline:end_sampling} of Algorithm~\ref{alg:montecarlo_psi}
is used only for the calculation of $|\mathcal{P}(a)|$, we do not
have to know all intact edges but only those on a
path from $s$ to $t$. One way of improving the procedure is to
sample only edges that can be reached from $s$. This can be done
by keeping track of the set of nodes $\mathcal{L}$ that
are currently reachable from
$s$. At the beginning, these are the nodes on the initial path
$\gamma$. Then all edges going out of $\mathcal{L}$ will be
sampled, and for every appeared edge we add its endpoint to
$\mathcal{L}$. This procedure is repeated until there are
no edges outgoing from $\mathcal{L}$ that have not already been sampled.

Finally, the determination of $|\mathcal{P}(a)|$ in
line~\ref{algline:multiplicity} of Algorithm~\ref{alg:montecarlo_psi}
can be easily done by an
analogue technique to the one used to calculate $w(\Omega)$
where this time we work on the subgraph of $G$ over the edges in
$a$ where we give a weight of one to every edge.

\section{Number of samples to draw}\label{sec:nb_samples}

In this section we analyze how many samples have to be drawn
in our algorithm to obtain an $\epsilon-\delta$ approximation.

\subsection{A priori bounds}

Using Theorem~\ref{thm:generalized_01_thm} we can derive the following bound.

\begin{theorem}\label{thm:montecarlo_importance_approx}
If the number of samples $N$ satisfies
\begin{equation*}
N\geq 4(e-2)\ln\left(\frac{2}{\delta}\right) \frac{w(\Omega)}{\epsilon^2 w(A)}
\end{equation*}
then $w(\Omega)\cdot\xi_N$ and $w(\Omega)\cdot \psi_N$ are both
$\epsilon-\delta$ approximations of $\mathit{REL}_{s,t}(G)$.
\end{theorem}

In Subsection~\ref{subsec:ratio_bounds}, we discuss upper bounds
for the ratio $w(\Omega)/w(A)$ which allow to apply
Theorem~\ref{thm:montecarlo_importance_approx} in practice and
to formulate conditions a network has to satisfy under which
Algorithm~\ref{alg:montecarlo_psi} is an FPRAS for
estimating $\mathit{REL}_{s,t}(G)$.

\subsection{Stopping criteria}

In practice, the use of a stopping
criterion typically allows to reduce the number of samples to
take as
we can profit from information gained during the execution
of the algorithm. Furthermore, we do not suffer from a possibly
weak upper bound for $w(A)/w(\Omega)$.
In \cite{dagum_2000_optimal}, Dagum et al. present two stopping criteria,
applicable to our algorithm and ensuring that
the result of the algorithm is an $\epsilon-\delta$ approximation
of $\mathrm{REL}_{s,t}(G)$. In the computational results, the second
algorithm (named \textit{Approximation Algorithm $\mathcal{AA}$}) is
used.

\subsection{Bounding the ratio \texorpdfstring{{\mathversion{bold}$w(\Omega)/w(A)$}}{$w(\Omega)/w(A)$}}\label{subsec:ratio_bounds}

In this subsection, we discuss upper bounds on the ratio
$w(\Omega)/w(A)$. Together with Theorem~\ref{thm:montecarlo_importance_approx},
these bounds allow us to bound the
number of iterations needed for our algorithm to deliver an
$\epsilon-\delta$ approximation.
The bounds discussed in this
section are correct not only for the case of directed acyclic
reliability networks but also for the more general case
of arbitrary directed, undirected or
even mixed reliability networks.

\paragraph{Previous results}

Karp an Luby \cite{karp_1985_montecarlo} gave the following upper bound on
the ratio $w(\Omega)/w(A)$ (in a slightly different context).

\begin{theorem}\label{thm:karp_bound_multiplicities}
\begin{equation*}
\frac{w(\Omega)}{w(A)}\leq \prod_{e\in E}(1+q(e))\;.
\end{equation*}
\end{theorem}

When combining the above theorem with
Theorem~\ref{thm:montecarlo_importance_approx} we obtain that
Algorithm~\ref{alg:montecarlo_psi} is an FPRAS if
$\prod_{e\in E}(1+q(e))$ is bounded by a polynomial in the input size
of the reliability network $G$.
In the special case of uniform edge failure probabilities,
i.e., $p(e)=\overline{p}= 1-\overline{q} \;\forall e\in E$,
Theorem~\ref{thm:karp_bound_multiplicities} reduces to
\begin{equation}\label{eq:karp_bound_multiplicities_uniform}
\frac{w(\Omega)}{w(A)}\leq (1+\overline{q})^m
\end{equation}
and implies that if
\begin{equation}\label{eq:karp_bound_q}
\overline{q}=O(\log(m)/m)
\end{equation}
then Algorithm~\ref{alg:montecarlo_psi} is an FPRAS for
estimating $\mathit{REL}_{s,t}(G)$.

\paragraph{Improved bound in the case
  of uniform edge failure probability}\label{sec:bound}

We now give a new bound on
$w(\Omega)/w(A)$ for the case of uniform edge failure probabilities,
which is sharper than (\ref{eq:karp_bound_multiplicities_uniform}),
especially in the case when the reliability network $G$ is
not too dense and does not contain long paths from $s$
to $t$.

To quantify the sparsity
of a graph we introduce the notion of \textit{\mbox{edge-vertex} bound}.
We say that a graph $G=(V,E)$
has an \mbox{edge-vertex} bound of $\mu$ if for any subset of the
vertices $U\subseteq V$ we have that there are at most $\mu
|U|$ edges in the subgraph of $G$ induced by the vertices
$U$. The best \mbox{edge-vertex} bound of a graph can be determined
in polynomial time by reduction to a flow problem \cite{gallo_1989_fast}. Furthermore,
$\Delta_{max}/2$, where $\Delta_{max}$ is the maximum degree of the
vertices in $G$, can be used as a simple valid \mbox{edge-vertex} bound.
Our new bound is given by the following theorem (see Appendix for a proof).

\begin{theorem}\label{thm:improved_bound}
Let $G=(V,E,\overline{p})$ be a reliability network with
uniform edge failure probability $\overline{p}=1-\overline{q}$,
\mbox{edge-vertex} bound $\mu$ and let $l_{min}$ respectively $l_{max}$ be the minimal and maximal length
of any $s$-$t$ path in $G$. Then we
have
\begin{equation*}
\frac{w(\Omega)}{w(A)}\leq%
\min\left\{\left(\frac{2}{2-\overline{q}}\right)^{m-l_{min}},%
2^{1+\frac{\mu}{m}\left[(\overline{q}m+%
\ln(2))(\overline{q}m+\ln(2)+l_{max})\right]}\right\}.
\end{equation*}
\end{theorem}

The first term of the minimum in our bound is a
slight improvement over the bound given by
(\ref{eq:karp_bound_multiplicities_uniform}). This can be
seen by observing that $2/(2-x)\leq (1+x) \;\forall x\in
[0,1]$. Contrary to bound
(\ref{eq:karp_bound_multiplicities_uniform}), the first term
of our bound may
be sharp. Furthermore it is independent of the graph topology and can easily
be generalized to \mbox{non-uniform} failure probabilities.

The second term of the minimum in
Theorem~\ref{thm:improved_bound} tries to exploit some
structure of the underlying network and is particularly interesting
for graphs with a low \mbox{edge-vertex} bound $\mu$ and without
long paths. For example, when working with networks where $\mu$ is bounded
by a constant and $l_{max}=O(\sqrt{m\log(m)})$,
Theorem~\ref{thm:improved_bound} implies that
Algorithm~\ref{alg:montecarlo_psi} is an FPRAS when
$\overline{q}=O(\sqrt{\log(m)/m})$,
which is a much weaker condition than the bound given by (\ref{eq:karp_bound_q}).

\section{Computational results}\label{sec:computational_results}

In order to test our algorithm we used two random generators for
creating directed acyclic graphs with low reliability. These generators are introduced in
the first part of this section. In a second part, we analyze the
running time of our algorithm on networks created by these
generators with different sizes and reliabilities.
Furthermore, the proposed algorithm is compared to a
direct \mbox{Monte-Carlo} simulation.

\subsection{Test instances}

\paragraph{Delaunay graphs (DEL)}

Our generator for directed acyclic Delaunay graphs takes two parameters, the number of
vertices $n$ and a uniform edge intactness probability $\overline{q}$.
We begin by choosing $n$ points uniformly at random in the unit square and consider the
undirected graph given by a Delaunay triangulation of these points. The two
terminals $s$ and $t$ are chosen as two vertices with maximal Euclidian
distance. We give a linear orientation to the edges corresponding to
the vector from $s$ to $t$, i.e., an (undirected) edge $\{v,w\}$ is oriented
as $(v,w)$ if the vector from $v$ to $w$ and the one from $s$ to $t$ have
a \mbox{non-negative} scalar product, otherwise we take the orientation $(w,v)$.
Finally, all edges get uniform intactness probability equal to $\overline{q}$.
One can easily observe that this construction guarantees that every vertex lies
on a path from $s$ to $t$.

\paragraph{Topological construction (TC)}

A second generator we use has three parameters, the number of vertices
$n$, a density parameter $\lambda\in[0,1]$ allowing to control the
expected number of edges in the graph and a parameter $\alpha$ influencing
the intactness probabilities. We begin with an empty
graph over $n$ vertices $V=\{v_1,v_2,\dots,v_n\}$ where $v_1=s$ and $v_n=t$.
The graph will be constructed such that $(v_1,v_2,\dots,v_n)$ is a topological
order of the vertices. In a first step, all edges of the graph are introduced,
then intactness probabilities are assigned to the edges.

For $i\in\{1,2,\dots,n-1\}$ we introduce an edge from $v_i$ to $v_{i+1}$. This
ensures that all vertices are on a path from $s$ to $t$.
All other possible edges will be present
with probability $\lambda$, i.e., for every $i,j\in\{1,2,\dots,n\}$ with $i+2\leq j$
we add an edge $(v_i,v_j)$ with probability $\lambda$.

Finally, the intactness probability of an edge $(v_i,v_j)$ is a number chosen
uniformly at random in the interval $[0,\frac{1}{(j-i)^{1-\alpha}}]$. By choosing
$\alpha < 1$, edges connecting topological near vertices have in general higher
intactness probabilities than edges connecting vertices being far away from each
other in the topological order. Therefore, smaller values for $\alpha$
result in less reliable networks. The value of $\alpha$ will typically be
chosen in $[0,1]$.

\paragraph{Parameter choice}

The parameters were fixed in such a way that networks of different reliabilities
were obtained for graphs with $10^3$,$10^4$,$10^5$ and $10^6$ vertices.
We generated DEL instances
for every
$\overline{q}\in\{0.01,0.02,\dots,1\}$ for every graph size $n\in\{10^3,10^4,10^5,10^6\}$.
For the creation of TC instances
the parameter $\lambda$ was always chosen such that the expected degree of
every vertex is equal to ten. This ensures that all graphs generated with the TC
generator having the same number of vertices also have about the same number of edges
and simplifies the comparison of running times. For every graph size $n\in\{10^3,10^4,10^5,10^6\}$
instances were generated for $\alpha\in\{0.01,0.02,\dots,1\}$.

The computational results presented in this section have been obtained on workstations
equipped with an AMD processor 3200+ and 1GB of RAM.

\subsection{Results and interpretations}

As a first observation, we noticed that the dependence of the running time
on $\epsilon$ and $\delta$ is essentially proportional to
$\log(\frac{1}{\delta})/{\epsilon^2}$ as predicted. We therefore
fixed $\epsilon=0.1$, $\delta=0.001$ for all results presented in this section.

Figure~\ref{fig:relrun} shows the running time
of the proposed algorithm as a function of the estimated reliability. As
expected, the running time grows when larger reliabilities have to be estimated,
as an intact state contains often several paths from $s$ to $t$. TC instances
with about the same reliability as DEL instances are much easier to tackle.
This comes from the fact that DEL instances are rather locally connected,
whereas most edges in TC instances were randomly chosen. Local connectedness
has the effect that when having an intact path $P$ from $s$ to $t$, there is
a good chance that several small subpaths of $P$ can be replaced by other intact
subpaths with the same start and endpoint. Every subpath which can be replaced
in this way raises the number of intact paths from $s$ to $t$. Furthermore, as
the replaceable subpaths are small it is likely that large groups of them are
disjoint, implying that various combinations of these subpath replacements
yield new intact paths from $s$ to $t$.
Figure~\ref{fig:relrun} shows that
even instances with $10^6$ vertices could be solved in reasonable time
as long as the estimated reliability was not too large.

\begin{figure}[h!]
\begin{center}
\resizebox{\linewidth}{!}{\includegraphics{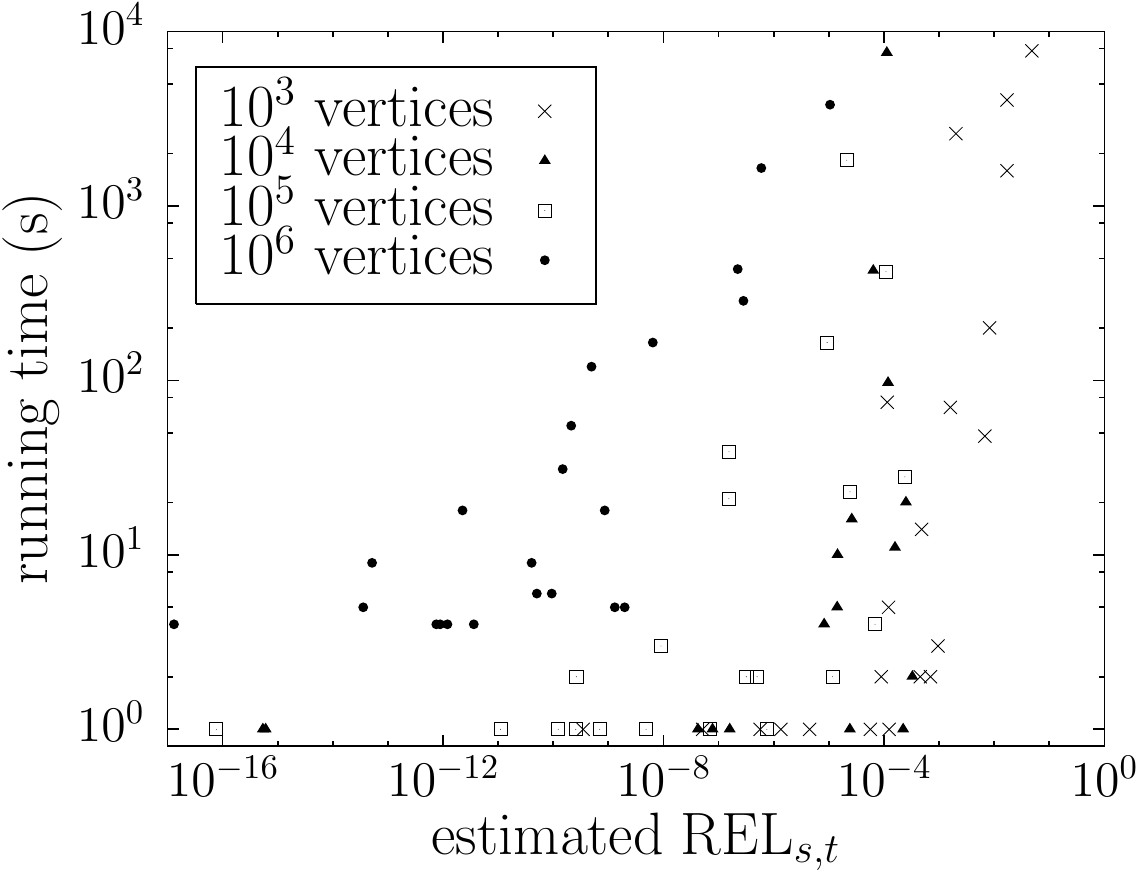}\hspace*{5mm}\includegraphics{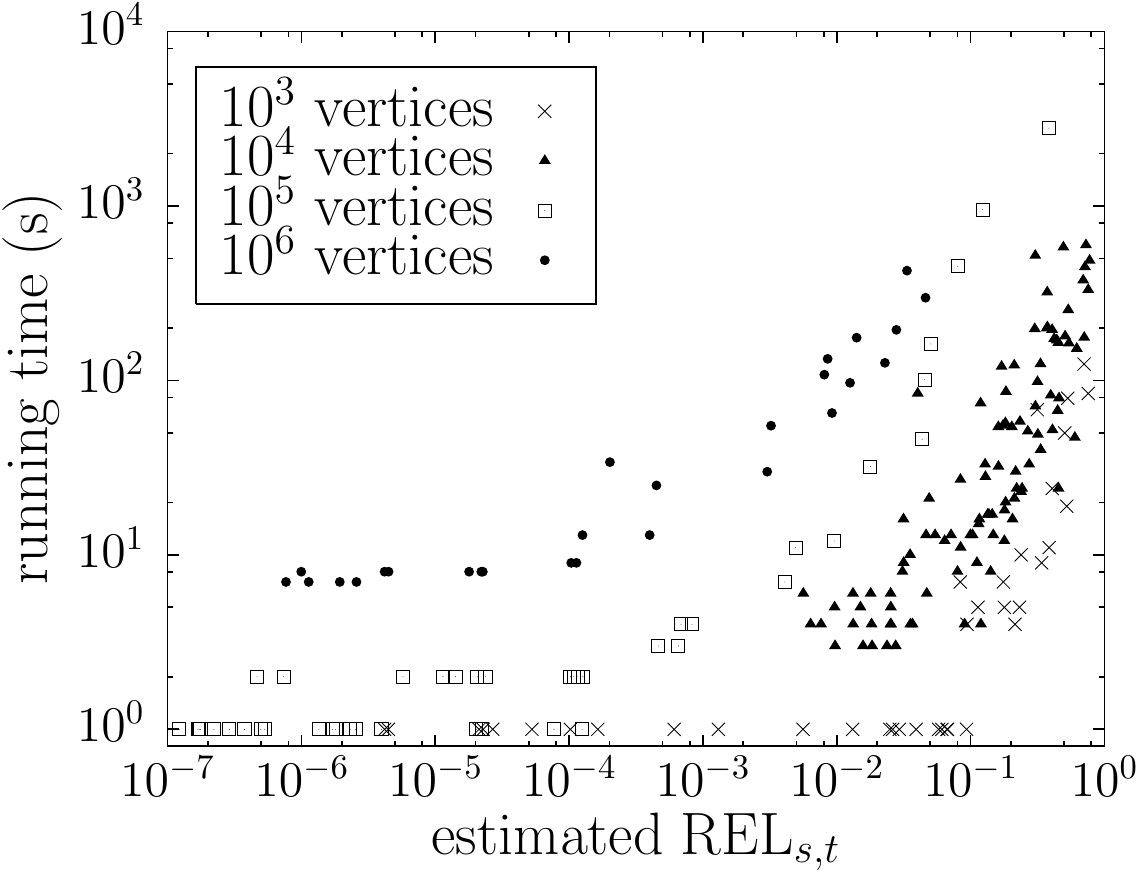}}
\caption{\label{fig:relrun} Running times of the proposed algorithm in
  function of the estimated reliabilities for DEL instances (on
  the left) and TC instances (on the right).}
\end{center}
\end{figure}

For the comparison of the proposed \mbox{Monte-Carlo} algorithm with
a direct \mbox{Monte-Carlo} approach, instances with $1000$ vertices
were used as most of these instances could have been solved in
reasonable time by both algorithms. As well as the proposed
\mbox{Monte-Carlo} algorithm, the direct \mbox{Monte-Carlo}
algorithm was implemented by using the sampling technique explained
in \ref{subsec:sampling_remaining_edges} allowing to reduce the
time needed per iteration in most instances.
Figure~\ref{fig:1000_compare} shows
the running times of both algorithms on DEL and TC instances with $1000$
vertices. As expected, the direct \mbox{Monte-Carlo} approach has an
approximately linear dependence on the reciprocal of the
estimated reliability.
Figure~\ref{fig:1000_compare} shows
the strength of the proposed algorithm when low reliabilities
have to be estimated.

\begin{figure}[h!]
\begin{center}
\resizebox{\linewidth}{!}{\includegraphics{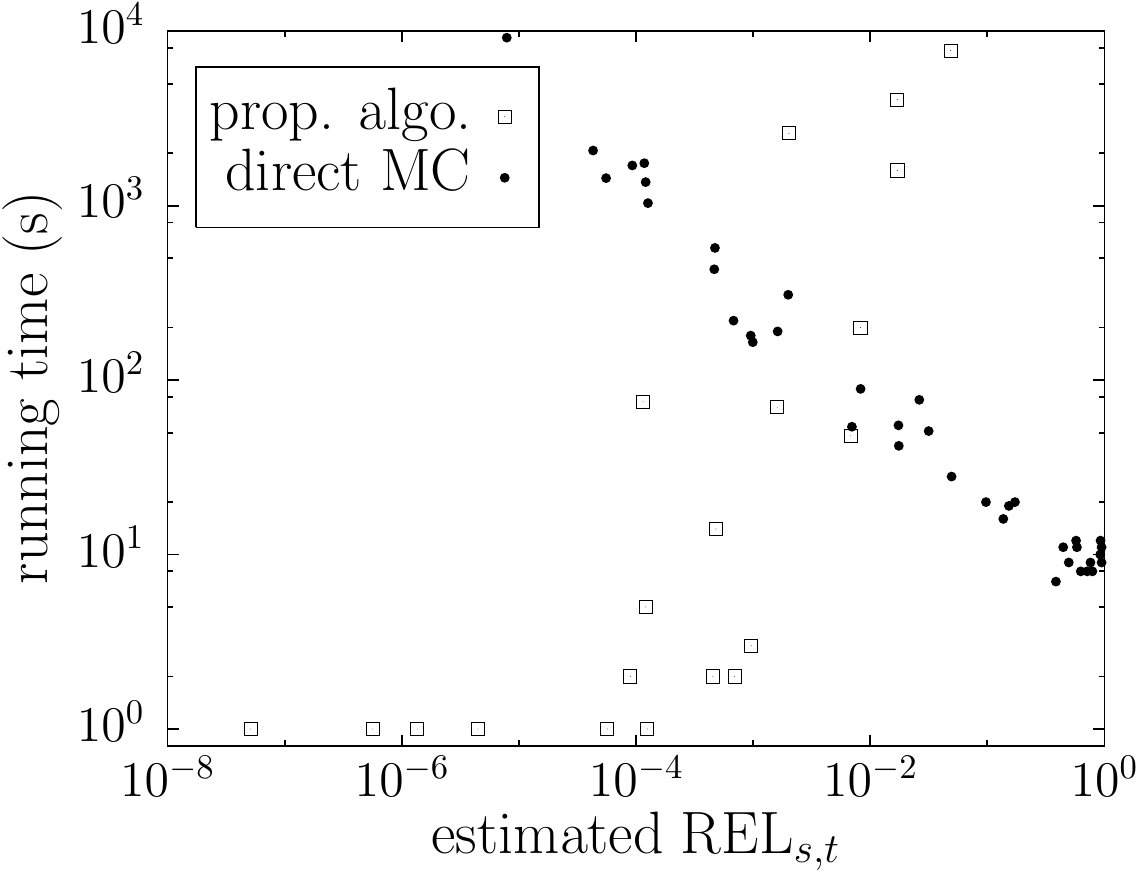}%
\hspace*{5mm}\includegraphics{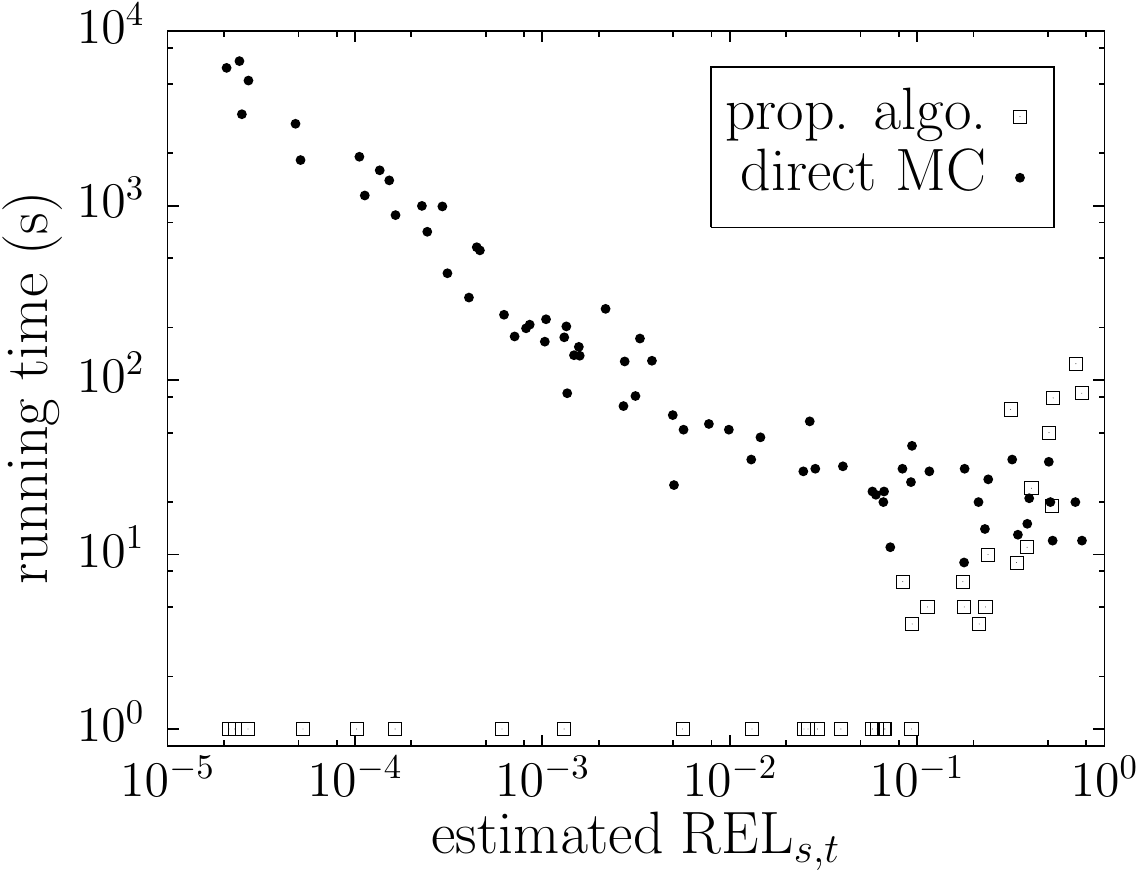}}
\caption{\label{fig:1000_compare} Running times of the proposed
algorithm (prop. algo.) and the direct \mbox{Monte-Carlo} (MC) approach
for DEL instances (on the left) and TC instances (on the right) with 
$1000$ vertices.}
\end{center}
\end{figure}

As the running time of the direct \mbox{Monte-Carlo} approach
depends nearly linear in $1/REL_{s,t}=1/w(A)$ (see
Figure~\ref{fig:1000_compare}) and
we expect that the proposed algorithm has a running time approximately
proportional to $w(\Omega)/w(A)$ (i.e. the reciprocal of the value
to estimate), we
expect the ratio between the running time of the proposed algorithm
and the direct \mbox{Monte-Carlo} approach is approximately linear
in $w(\Omega)$. This is confirmed by Figure~\ref{fig:1000_wOmegarunratio}
showing the ratio of the running time of both algorithms in function
of $w(\Omega)$ for DEL and TC instances with $1000$ vertices.
It is not surprising that both algorithms need about the same running time
when $w(\Omega)$ is near to one. This observation allows to perform a
simple a priori test for deciding which algorithm is better suited
for a particular instance. Given an instance, we first calculate
$w(\Omega)$ (in linear time). If $w(\Omega) \ll 1$ it is likely that
the proposed algorithm will be faster than the direct \mbox{Monte-Carlo}
approach. When $w(\Omega) \gg 1$, the direct \mbox{Monte-Carlo} approach
is likely to be the more efficient algorithm.

\begin{figure}[ht!]
\begin{center}
\includegraphics[width=0.6\linewidth]{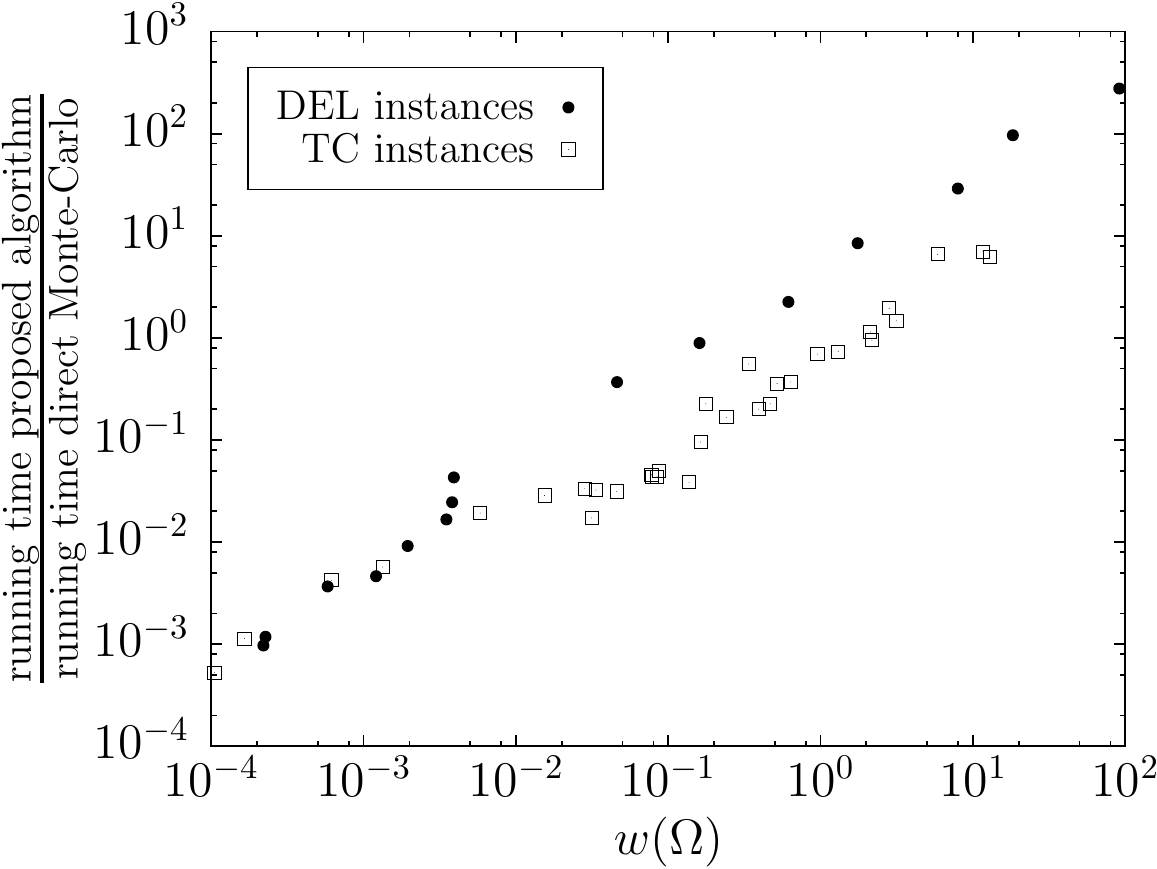}
\caption{\label{fig:1000_wOmegarunratio} Comparison of the running time of the proposed
algorithm and the direct \mbox{Monte-Carlo} approach in function of $w(\Omega)$
for networks with $1000$ vertices.}
\end{center}
\end{figure}

\section{Conclusions}

An adapted version of the Monte-Carlo algorithm given by Karp and Luby
in \cite{karp_1985_montecarlo} was presented and analyzed.
The new algorithm is specialized for directed acyclic graphs and is suited for the estimation
of small reliabilities. Computational results show the successful
application of the proposed algorithm on two types of randomly
generated large-scale instances 
and its advantage compared to the direct \mbox{Monte-Carlo} approach
when very small reliabilities have to be estimated.
Previous algorithms for accurate estimation of $s$-$t$ reliability
were only applicable on either very small instances or on
a very restricted class of initial networks.
For the case of uniform
edge failure probabilities, a \mbox{worst-case} bound
on the number of samples to be drawn was given that sharpens a bound presented in
\cite{karp_1985_montecarlo} and is significantly stronger in the case of
relatively sparse graphs without long paths from $s$ to $t$.

One important open question in this domain is if there exists an FPRAS
for estimating $s$-$t$ reliability in directed acyclic graphs. It would be interesting to
find algorithms allowing to tackle instances
efficiently that cannot be solved in reasonable time by our algorithm or the
direct \mbox{Monte-Carlo} approach.
Another point is the generalization of the upper bound for $w(\Omega)/w(A)$
for the case of \mbox{non-uniform} failure probabilities.
Additionally for the case of general (not necessarily acyclic) networks
there seem to be no practically efficient algorithms at the moment for the
estimation of low reliabilities on large instances.

\bibliographystyle{plain}
\bibliography{networks}

\section*{Appendix}

\subsubsection*{Proof of Theorem~\ref{thm:improved_bound}}

Let $(\gamma,a) \in \Omega$ be a
sample drawn
according to lines
\mbox{\ref{algline:path_sampling}-\ref{algline:end_sampling}}
of \mbox{Algorithm~\ref{alg:montecarlo_psi}}. We define the
influence value $m(\gamma,a)$ of the sample as in
Section~\ref{sec:our_algorithm} as a random variable
\begin{equation*}
m: \Omega \longrightarrow \Real, \quad m(\gamma,a)=\frac{1}{|\mathcal{P}(a)|}\;.
\end{equation*}

As discussed in Section~\ref{sec:our_algorithm}, the influence
value of a sample from $\Omega$ can be used as an unbiased estimator for
$w(A)/w(\Omega)$ as we have
\begin{equation*}
\ex[m(\gamma,a)]=\frac{w(A)}{w(\Omega)}\;.
\end{equation*}

By the above equation, the reciprocal of a lower bound on $\ex[m]$
is an upper bound on $w(\Omega)/w(A)$. We will therefore
deduce the bound given in Theorem~\ref{thm:improved_bound} by
deriving a lower bound bound on $w(A)/w(\Omega)$.

Let $\Omega_{l,k}$ be the subset of all elements of the sample
space $\Omega$ where the initial chosen path has length $l$
and the total number of appeared edges is $l+k$, i.e.,

\begin{equation*}
\Omega_{l,k}=\{(\gamma,a)\in\Omega\mid |\gamma|=l, |a|=l+k\}\;.
\end{equation*}

It is clear that all elements $\omega \in \Omega_{l,k}$ appear with equal probability
and can be
sampled by the method described in Algorithm~\ref{alg:uniform_Omegalk}. This method is introduced just
for theoretical analysis.

\linesnumbered
\begin{algorithm}[H]
\SetLine
\caption{Sampling uniformly from $\Omega_{l,k}$\label{alg:uniform_Omegalk}}
Choose uniformly a path $\gamma\in\mathcal{P}$ with length $l$\;
$a=\gamma$\;
$R=E\setminus \gamma$\;
\For{$i=1$ to $k$}{Choose an edge $e_i\in R$ uniformly
  at random\;
$a=a\cup e_i$\;
$R=R\setminus e_i$\;}
\Return{$(\gamma,a)$}
\end{algorithm}

With every edge $e_i$ added in the \mbox{for-loop} of
Algorithm~\ref{alg:uniform_Omegalk} we associate a multiplicity
$M_{e_i}\in\{1/2,1\}$ equal to $1/2$ if both endpoints of $e_i$
are saturated by edges in $\gamma \cup \{e_j\mid 1\leq j \leq i-1\}$
and equal to $1$ otherwise. Intuitively, the multiplicity
is a measure for the influence of an added
edge in Algorithm~\ref{alg:uniform_Omegalk} on the ratio
$w(A)/w(\Omega)$. The following lemma formalizes this
intuition.

\begin{lemma}\label{lem:multiplicity}
Let $(\gamma,a)\in \Omega_{l,k}$ be an intact state constructed as described in
Algorithm~\ref{alg:uniform_Omegalk}. Then we have
\begin{equation*}
\frac{1}{|\mathcal{P}(a)|}\geq \prod_{i=1}^k M_{e_i}\;.
\end{equation*}
\end{lemma}

Before proving this lemma, we discuss the link between the lemma
and Algorithm~\ref{alg:uniform_Omegalk}. In general there
are multiple ways to get an intact state $a\in A$ with
Algorithm~\ref{alg:uniform_Omegalk}. Depending on how the
state $a$ was obtained, the multiplicities associated with
the edges in $a$ are different. As
Lemma~\ref{lem:multiplicity} is true for every possible way
of obtaining state $a$ by Algorithm~\ref{alg:uniform_Omegalk}, it
is applicable even to intact states that were not
constructed through Algorithm~\ref{alg:uniform_Omegalk}. One
just has to fix a possible way how the state $a$ could have
been constructed by Algorithm~\ref{alg:uniform_Omegalk},
i.e., an $s$-$t$ path $\gamma\subseteq a$ has to be fixed as well
as an order for the edges in $a\setminus\gamma$ specifying in
which sequence those edges were chosen in
Algorithm~\ref{alg:uniform_Omegalk}.
The multiplicities can then be calculated with respect to this order, and
Lemma~\ref{lem:multiplicity} can be applied.

\begin{proof}[Proof of Lemma~\ref{lem:multiplicity}]
Let $(\gamma,a)\in \Omega_{l,k}$ be an intact state
constructed as described in
Algorithm~\ref{alg:uniform_Omegalk}. We define
$F=\{e_1,e_2,\dots,e_k\}=a\setminus\gamma$ to be the set of
all edges added to the initial path $\gamma$ during the
construction of $a$. Let $F^{\frac{1}{2}},F^{1}$ be
the following partitioning of the edges in $a$.
\begin{align*}
F^{\frac{1}{2}}&=\{e\in F \mid M_{e}=\frac{1}{2}\}\\
F^1&=\gamma \cup \{e\in F \mid M_{e}=1\}
\end{align*}
We prove the following statement, which immediately
implies Lemma~\ref{lem:multiplicity}.

\begin{proposition*}
For every set $H\subseteq F^{\frac{1}{2}}$ there exists at
most one $s$-$t$ path in the state $a$ that contains all
edges of $H$ and none of $F^\frac{1}{2}\setminus H$.
\end{proposition*}

Lemma~\ref{lem:multiplicity} follows from the above
proposition by the following observation. The proposition
implies that there are at most as many different $s$-$t$
paths in $a$ as there are subsets of $F^{\frac{1}{2}}$. We
therefore have
\begin{equation*}
|\mathcal{P}(a)|\leq 2^{\left|F^{\frac{1}{2}}\right|},
\end{equation*}
which finally implies
\begin{equation*}
\frac{1}{|\mathcal{P}(a)|}\geq%
\left(\frac{1}{2}\right)^{\left|F^{\frac{1}{2}}\right|}=\prod_{i=1}^k M_{e_i}\;.
\end{equation*}

It remains to prove the proposition. Let $H\subseteq
F^{\frac{1}{2}}$ and suppose that we have two different
$s$-$t$ paths $\gamma_1,\gamma_2\subseteq a$ such that both
contain all edges of $H$ and none of
$F^{\frac{1}{2}}\setminus H$. This implies that their
symmetric difference $\gamma_1 \Delta \gamma_2\subseteq a$
contains a cycle consisting only of edges in $F^1$, thus
contradicting the fact that $F^1$ does not contain cycles
(any cycle in $a$ contains at least one element of $F^{\frac{1}{2}}$).\\
\end{proof}

With the aid of Lemma~\ref{lem:multiplicity} we prove
the following intermediate result.

\begin{lemma}\label{lem:bound_Omegalk}
Let $l\in \mathbb{N}, k\in \{0\}\cup\mathbb{N}$
with $l+k\leq m$ and $\omega=(\gamma,a)$ be
a random sample from the sample space
according to Algorithm~\ref{alg:montecarlo_psi}. We have
\begin{equation*}
\ex\left[\frac{1}{|\mathcal{P}(a)|} \mid
  (\gamma,a)\in\Omega_{l,k} \right]%
\geq \max\left\{2^{-k},2^{-\frac{\mu}{m}k(k+l)}\right\}\;.
\end{equation*}
\end{lemma}

\begin{proof}[Proof of Lemma~\ref{lem:bound_Omegalk}]
Let $\omega'=(\gamma',a')$ be a sample corresponding
to a result of Algorithm~\ref{alg:uniform_Omegalk} for the
given $l$ and $k$.
Conditioned on $\Omega_{l,k}$, $\omega$ has therefore the same
distribution as $\omega'$.
Let $M_{e_1},M_{e_2},\dots,M_{e_k}$ be the random variables
corresponding to the multiplicities of the edges in
$a'\setminus\gamma'$. By Lemma~\ref{lem:multiplicity} we
have
\begin{equation}\label{eq:ex_multiplicities}
\ex\left[\frac{1}{|\mathcal{P}(a)|} \mid
  (\gamma,a)\in\Omega_{l,k} \right]%
=\ex\left[\frac{1}{|\mathcal{P}(a')|}\right]
\geq \ex\left[\prod_{i=1}^k M_{e_i}\right]\;.
\end{equation}

Let $i\in \{1,\dots,k\}$. Observe that independently of which
path $\gamma'$ and which edges $e_{1},\dots,e_{k}$ were
chosen, we have that at most $l+1+2(i-1)$ vertices in
$G$ are saturated by the edges in $\gamma' \cup
\{e_{1},\dots,e_{i-1}\}$ ($l+1$ vertices are saturated
through $\gamma'$ and every additional edge saturates at most two
new vertices in $G$). Let $V_i$ be the vertices in $G$ which
are saturated by $\gamma' \cup \{e_{1},\dots,e_{i-1}\}$.

In Algorithm~\ref{alg:uniform_Omegalk}, the edge $e_i$ is chosen
uniformly at random from the remaining edges $E\setminus
(\gamma' \cup \{e_{1},\dots,e_k\})$. As $G$ is sparse
with \mbox{edge-vertex} bound $\mu$ we have that at most $\mu
|V^i|\leq \mu (l+1+2(i-1))$ edges have both
endpoints in $V_i$. Furthermore, $i-1$ of these edges were already
chosen. We therefore have the following stochastic
inequality (which is true for any realization of
$\gamma',M_{e_{1}},\dots,M_{e_{i-1}}$):
\begin{align*}
\prob\left(M_{e_i}=\frac{1}{2} \mid \gamma',
M_{e_{1}},\dots,M_{e_{i-1}}\right) &\leq%
\min\left\{1,\frac{\mu |V^i|-(i-1)}{m-(i-1)}\right\}\\
&\leq \min\left\{1,\frac{\mu |V_i|}{m}\right\}\\
&\leq \min\left\{1,\frac{\mu(2i+l-1)}{m}\right\}\;.
\end{align*}

The stochastic inequality above allows to give a simple
bound on the following conditional expectation:
\begin{equation*}
\begin{aligned}
\ex[M_{e_i}\mid \gamma',M_{e_{1}},\dots,M_{e_{i-1}}]
&=&&\frac{1}{2} \, \prob\left(M_{e_i}=\frac{1}{2} \mid \gamma',M_{e_1},\dots,M_{e_{i-1}}\right)+\\
&&&\prob\left(M_{e_i}=1 \mid \gamma',M_{e_{1}},\dots,M_{e_{i-1}}\right)\\
&=&&1-\frac{1}{2}\prob\left(M_{e_i}=\frac{1}{2} \mid \gamma', M_{e_1},\dots,M_{e_{i-1}}\right)\\
&\geq&& \max\left\{\frac{1}{2},1-\frac{\mu(2i+l-1)}{2m}\right\}\;.
\end{aligned}
\end{equation*}

Applying the above inequality, the expectation in
(\ref{eq:ex_multiplicities}) can be developed as
\begin{align*}
\ex\left[\prod_{i=1}^k M_{e_i}\right] &\geq
\prod_{i=1}^k\max\left\{\frac{1}{2},1-\frac{\mu(2i+l-1)}{2m}\right\}\\
&\geq \prod_{i=1}^k\max\left\{\frac{1}{2},%
\left(\frac{1}{4}\right)^{\frac{\mu(2i+l-1)}{2m}}\right\}\;.
\end{align*}
The last inequality comes from the fact that $(1/4)^x
\leq 1-x$ for $x\in[0,1/2]$ and $(1/4)^x\leq 1/2$ for $x\geq
1/2$.
Developing further, we finally get the result of
Lemma~\ref{lem:bound_Omegalk}.
\begin{align*}
\prod_{i=1}^k\max\left\{\frac{1}{2},%
\left(\frac{1}{4}\right)^{\frac{\mu(2i+l-1)}{2m}}\right\}
&\geq%
\max\left\{2^{-k},2^{-\overset{k}{\underset{i=1}{\sum}}\frac{\mu(2i+l-1)}{m}}\right\}\\
&\geq%
\max\left\{2^{-k},2^{-\frac{\mu}{m}k(k+l)}\right\}
\end{align*}
\end{proof}

Beginning with the result of Lemma~\ref{lem:bound_Omegalk}
we now prove Theorem~\ref{thm:improved_bound} by first weakening and
then eliminating the conditioning on $\Omega_{l,k}$.
Let $\Omega_l$ be the set of all elements of the sample
space $\Omega$ where the initial chosen path has length
$l$, i.e.,
\begin{equation*}
\Omega_l=\{(\gamma,a)\mid |\gamma|=l\}\;.
\end{equation*}

Let $K$ be the random variable corresponding to the number
of edges that appeared additionally to the ones of the
initial path, when drawing an element out of $\Omega_l$.
Note that $K$ is binomially distributed as
\begin{equation*}
K\sim Bin(m-l,\overline{q})\, .
\end{equation*}

Using Lemma~\ref{lem:bound_Omegalk} we get
\begin{align}
\ex\left[\frac{1}{|\mathcal{P}(a)|}\mid (\gamma,a)\in\Omega_l
\right]%
&\geq%
\ex\left[\max\left\{2^{-K},2^{-\frac{\mu}{m}K(K+l)}\right\}\right]
\notag \\
&\geq%
\max\left\{\ex\left[2^{-K}\right],\ex\left[2^{-\frac{\mu}{m}K(K+l)}\right]\right\}
\notag\\
&=%
\max\left\{\left(1-\frac{\overline{q}}{2}\right)^{m-l},%
\ex\left[2^{-\frac{\mu}{m}K(K+l)}\right]\right\}\;.\label{eq:bound_Omegal}
\end{align}

By replacing $l$ by $l_\mathrm{min}$ in the first term of the
above maximum we get the first part of the inequality in Theorem~\ref{thm:improved_bound}, i.e.,
\begin{equation*}
\frac{w(\Omega)}{w(A)}=\left(\ex\left[\frac{1}{|\mathcal{P}(a)|}\right]\right)^{-1}%
\leq%
\left(1-\frac{\overline{q}}{2}\right)^{-(m-l_\mathrm{min})}%
=%
\left(\frac{2}{2-\overline{q}}\right)^{m-l_\mathrm{min}}\;.
\end{equation*}

The remaining part of Theorem~\ref{thm:improved_bound} will be shown
by developing the second term of (\ref{eq:bound_Omegal}) further.
We will use the inequality
\begin{equation*}
\ex[2^{-\frac{\mu}{m}K(K+l)}]\geq
2^{-\frac{\mu}{m}\alpha(\alpha+l)} \prob[K\leq\alpha] \, ,
\end{equation*}
which is true for any $\alpha \in \mathbb{R}$ and a result shown by
Hamza \cite{hamza_1995_smallest} stating that the distance between
the median and the mean of a binomial random variable is at most
$\ln(2)$.
Therefore, when choosing
$\alpha=\overline{q}m+\ln(2)$, we have
$\prob[K\leq\alpha]\geq 1/2$ and get the following result:
\begin{align*}
\ex[2^{-\frac{\mu}{m}K(K+l)}]%
&\geq%
2^{-1-\frac{\mu}{m}(\overline{q}(m-l)+\ln(2))(\overline{q}(m-l)+\ln(2)+l)}\\
&\geq%
2^{-1-\frac{\mu}{m}(\overline{q}m+\ln(2))(\overline{q}m+\ln(2)+l)}\\
&\geq%
2^{-1-\frac{\mu}{m}(\overline{q}m+\ln(2))(\overline{q}m+\ln(2)+l_\mathrm{max})} .
\end{align*}
This implies the second part of the inequality in
Theorem~\ref{thm:improved_bound} as
\begin{equation*}
\frac{w(\Omega)}{w(A)}=\left(\ex\left[\frac{1}{|\mathcal{P}(a)|}\right]\right)^{-1}%
\leq%
2^{1+\frac{\mu}{m}(\overline{q}m+\ln(2))(\overline{q}m+\ln(2)+l_\mathrm{max})}\;,
\end{equation*}
which completes the proof of Theorem~\ref{thm:improved_bound}.

\end{document}